%% file: paper.tex
\documentclass[a4paper,UKenglish,cleveref, autoref, thm-restate]{lipics-v2021} %anonymous
\pdfoutput=1 %uncomment to ensure pdflatex processing (mandatatory e.g. to submit to arXiv)
\hideLIPIcs  %uncomment to remove references to LIPIcs series (logo, DOI, ...), e.g. when preparing a pre-final version to be uploaded to arXiv or another public repository

\usepackage[references]{agda}
\usepackage{tikz-cd}
\usepackage{quiver}
\usepackage{stmaryrd} %for Dana Scott brackets
\usepackage{mathtools} %for big brackets
\usepackage[textwidth=30mm]{todonotes} %for to-do notes
\usepackage{newunicodechar} %for Agda unicode rendering
\usepackage{amsmath} %for \eqref
\usepackage{tikz} %for diagrams
\usepackage{tikz-cd} %for commutative diagrams
\usepackage{xcolor} %for colours
\usepackage{xspace} %for system names
\usepackage{xargs} % Use more than one optional parameter in a new commands
\usepackage{adjustbox} %for resizing tikzcd
\usepackage{catchfilebetweentags}
\usepackage{cite} %for multiple citations

\input{macros.tex}

\title{Formalising Inductive and Coinductive Containers} %TODO Please add

\author{Stefania Damato}{School of Computer Science, University of Nottingham, UK \and \url{https://stefaniatadama.com}}{stefania.damato@nottingham.ac.uk}{https://orcid.org/0009-0001-7182-5304}{} %{[orcid]}{[funding]} %TODO mandatory, please use full name; only 1 author per \author macro; first two parameters are mandatory, other parameters can be empty. Please provide at least the name of the affiliation and the country. The full address is optional. Use additional curly braces to indicate the correct name splitting when the last name consists of multiple name parts.

\author{Thorsten Altenkirch}{School of Computer Science, University of Nottingham, UK \and \url{https://www.cs.nott.ac.uk/~psztxa/}}{thorsten.altenkirch@nottingham.ac.uk}{}{}%{[orcid]}{[funding]}

\author{Axel Ljungström}{Department of Mathematics, Stockholm University, Sweden \and \url{https://aljungstrom.github.io}}{axel.ljungstrom@math.su.se}{https://orcid.org/0000-0001-6946-0775}{This author is supported by the Swedish Research Council (Vetenskapsrådet)
under Grant No. 2019-04545.} 

\authorrunning{S. Damato, T. Altenkirch, and A. Ljungström} %TODO mandatory. First: Use abbreviated first/middle names. Second (only in severe cases): Use first author plus 'et al.'

\Copyright{Stefania Damato, Thorsten Altenkirch, and Axel Ljungström} %TODO mandatory, please use full first names. LIPIcs license is "CC-BY";  http://creativecommons.org/licenses/by/3.0/

 %TODO mandatory; please add comma-separated list of keywords

\relatedversion{} %optional, e.g. full version hosted on arXiv, HAL, or other respository/website
\acknowledgements{The authors would like to thank Anders Mörtberg, the anonymous reviewers, and the Agda Zulip and Discord communities for useful comments and discussions.}%optional

\nolinenumbers %uncomment to disable line numbering

\EventEditors{Yannick Forster and Chantal Keller}
\EventNoEds{2}
\EventLongTitle{16th International Conference on Interactive Theorem Proving (ITP 2025)}
\EventShortTitle{ITP 2025}
\EventAcronym{ITP}
\EventYear{2025}
\EventDate{September 29--October 2, 2025}
\EventLocation{Reykjavik, Iceland}
\EventLogo{}
\SeriesVolume{352}
\ArticleNo{17}
\begin{document}

\maketitle

\begin{abstract}
\label{sec-abstract}
\input{sections/abstract.tex}
\end{abstract}

\section{Introduction}
\label{sec-intro}
\input{sections/intro.tex}

\section{Background}
\label{sec-background}
\input{sections/background.tex}

\section{Setting up}
\label{sec-settingup}
\input{sections/settingup.tex}

\section{Fixed points}
\label{sec-fixedpoints}
\input{sections/fixedpoints.tex}

\section{Conclusion, Related Work, and Future Work}
\label{sec-conclusions}
\input{sections/conclusions.tex}

%%
%% Bibliography
%%

\bibliography{refs}

\end{document}

%% file: macros.tex
\newunicodechar{λ}{\ensuremath{\mathnormal\lambda}}
\newunicodechar{ℕ}{\ensuremath{\mathbb{N}}}
\newunicodechar{∧}{\ensuremath{\mathnormal\wedge}}
\newunicodechar{⇒}{\ensuremath{\mathnormal\Rightarrow}}
\newunicodechar{∨}{\ensuremath{\mathnormal\vee}}
\newunicodechar{→}{\ensuremath{\mathnormal\to}}
\newunicodechar{∀}{\ensuremath{\mathnormal\forall}}
\newunicodechar{₁}{\ensuremath{\mathnormal{_1}}}
\newunicodechar{₂}{\ensuremath{\mathnormal{_2}}}
\newunicodechar{₃}{\ensuremath{\mathnormal{_3}}}
\newunicodechar{₄}{\ensuremath{\mathnormal{_4}}}
\newunicodechar{₉}{\ensuremath{\mathnormal{_9}}}
\newunicodechar{⊤}{\ensuremath{\top}}
\newunicodechar{⊥}{\ensuremath{\perp}}
\newunicodechar{∷}{\ensuremath{\cln}}
\newunicodechar{ε}{\ensuremath{\varepsilon}}
\newunicodechar{∞}{\ensuremath{\infty}}
\newunicodechar{Π}{\ensuremath{\Pi}}
\newunicodechar{π}{\ensuremath{\pi}}
\newunicodechar{Σ}{\ensuremath{\Sigma}}
\newunicodechar{∈}{\ensuremath{\mathnormal\in}}
\newunicodechar{≡}{\ensuremath{\mathnormal\equiv}}
\newunicodechar{×}{\ensuremath{\mathnormal\times}}
\newunicodechar{⟦}{\ensuremath{\mathnormal\llbracket}}
\newunicodechar{⟧}{\ensuremath{\mathnormal\rrbracket}}
\newunicodechar{ʳ}{\ensuremath{^r}}
\newunicodechar{⊎}{\ensuremath{\uplus}}
\newunicodechar{∎}{\ensuremath{\blacksquare}}
\newunicodechar{⟨}{\ensuremath{\langle}}
\newunicodechar{⟩}{\ensuremath{\rangle}}
\newunicodechar{ℓ}{\ensuremath{\ell}}
\newunicodechar{α}{\ensuremath{\alpha}}
\newunicodechar{β}{\ensuremath{\beta}}
\newunicodechar{σ}{\ensuremath{\sigma}}
\newunicodechar{⋄}{\ensuremath{\diamond}}
\newunicodechar{ι}{\ensuremath{\iota}}
\newunicodechar{∘}{\ensuremath{\circ}}
\newunicodechar{⦉}{\ensuremath{\triangleleft}}
\newunicodechar{●}{\ensuremath{\mathbin{\vcenter{\hbox{\scalebox{0.6}{$\bullet$}}}}}}
\newunicodechar{γ}{\ensuremath{\gamma}}
\newunicodechar{⁻}{\ensuremath{^-}}
\newunicodechar{∙}{\ensuremath{\mathbin{\vcenter{\hbox{\scalebox{0.4}{$\bullet$}}}}}}
\newunicodechar{∫}{\ensuremath{\smallint}}
\newunicodechar{⊚}{\ensuremath{\circ}}
\newunicodechar{𝕓}{\ensuremath{\overline\beta}}
\newunicodechar{𝕥}{\ensuremath{\tilde{\beta}}}
\newunicodechar{⋆}{\ensuremath{\star}}
\newunicodechar{∼}{\ensuremath{\sim}}
\newunicodechar{◹}{}
\newunicodechar{◃}{}
\newunicodechar{≈}{\ensuremath{\approx}}
\newunicodechar{æ}{\ensuremath{\overline\alpha}}
\newunicodechar{♭}{\ensuremath{\tilde\beta}}
\newunicodechar{μ}{\ensuremath{\mu}}
\newunicodechar{ν}{\ensuremath{\nu}}
\newunicodechar{χ}{\ensuremath{\chi}}
\newunicodechar{ξ}{\ensuremath{\xi}}
\newunicodechar{ϕ}{\ensuremath{\phi}}
\newunicodechar{η}{\ensuremath{\eta}}
\newunicodechar{≅}{\ensuremath{\cong}}
\newunicodechar{₀}{\ensuremath{_0}}
\newunicodechar{₁}{\ensuremath{_1}}
\newunicodechar{¹}{\ensuremath{^1}}
\newunicodechar{≊}{\ensuremath{\approxeq}}
\newunicodechar{≈}{\ensuremath{\approx}}
\newunicodechar{≃}{\ensuremath{\simeq}}
\newunicodechar{⁈}{{\color{black}\ensuremath{\dots}}}
\newunicodechar{♥}{\ensuremath{\approxeq}}

\DeclarePairedDelimiter{\pa}{(}{)}
\DeclarePairedDelimiter{\pads}{\llbracket}{\rrbracket}

\definecolor{agdablue}{HTML}{0000CD}
\definecolor{agdagreen}{HTML}{008B00}
\definecolor{agdapink}{HTML}{EE1289}

\setuldepth{Set}
\newlength{\depthofsumsign}
\newcommand{\systemname}[1]{\texttt{\color{darkgray}#1}\xspace}
\newcommand{\cubicalagda}{\systemname{Cubical Agda}}
\newcommand{\agdaCubical}{\systemname{agda/cubical}}
\newcommand{\agda}{\systemname{Agda}}
\newcommand{\agdablue}[1]{\normalfont{\AgdaFunction{#1}}}
\newcommand{\agdagreen}[1]{\normalfont{\AgdaInductiveConstructor{#1}}}
\newcommand{\agdapink}[1]{\normalfont{\AgdaField{#1}}}
\newcommand{\agdaorange}[1]{\normalfont{\AgdaKeyword{#1}}}
\newcommand{\cat}[1]{\normalfont{\ul{\textbf{#1}}}}

\newcommand{\cont}[2]{#1\triangleleft #2}
\newcommand{\contfunc}[2]{\llbracket #1 \triangleleft #2 \rrbracket}

\newcommand{\contfuncX}[3]{\llbracket #1 \triangleleft #2 \rrbracket\,#3}
\newcommand{\bigcontfuncX}[3]{\pads[\Big]{#1 \triangleleft #2}\,#3}

\newcommand{\cln}{\ensuremath{\mathbin{\raisebox{.8pt}{\scalebox{0.85}{:\kern-.3pt:}}}}}
\newcommand{\danascott}{\llbracket\text{\textunderscore}\rrbracket}

\newcommand{\Fin}{\mathsf{Fin}}
\newcommand{\nsum}[1][1]{% only for \displaystyle
    \mathop{%
        \raisebox
        {-#1\depthofsumsign+1\depthofsumsign}
        {\scalebox
            {#1}
            {$\displaystyle\sum$}%
        }
    }
}
\newcommand{\nprod}[1][1]{% only for \displaystyle
    \mathop{%
        \raisebox
        {-#1\depthofsumsign+1\depthofsumsign}
        {\scalebox
            {#1}
            {$\displaystyle\prod$}%
        }
    }
}

\newcommand{\comm}[1]{\ensuremath{comm_{#1}}}

\newcommand{\minipageskip}{\vspace{-.3cm}}

\newcommand{\agdaeq}[2]{{#1}\;\agdablue{$\equiv$}\;{#2}}

\newcommand{\myarrow}[1][]{%
\raisebox{0.13 ex}{
  \begin{tikzpicture}[#1]%
    \draw (0,0.75ex) -- (0,-.25ex) -- (0.3em,-.25ex);
    \draw (-0.4ex,0.4ex) -- (0,0.75ex) -- (0.4ex,0.4ex);
  \end{tikzpicture}%
  }
}

\newcommand{\mycomment}[1]{
\begin{minipage}[t]{.1\textwidth}
  \hspace{1cm}{\myarrow}
\end{minipage}%
\colorbox{yellow!15}{
\begin{minipage}[t]{.85\textwidth}
{\scriptsize #1
}
\end{minipage}
}
\begin{minipage}[t]{.05\textwidth}
  \hfill
\end{minipage}
}

%% file: sections/abstract.tex
Containers capture the concept of strictly positive data types in programming. The original development of containers is done in the internal language of locally cartesian closed categories (LCCCs) with disjoint coproducts and W-types, and uniqueness of identity proofs (UIP) is implicitly assumed throughout. Although it is claimed that these developments can also be interpreted in extensional Martin-L\"of type theory, this interpretation is not made explicit. In this paper, we present a formalisation of the results that `containers preserve least and greatest fixed points' in \cubicalagda, thereby giving a formulation in intensional type theory. Our proofs do not make use of UIP and thereby generalise the original results from talking about container functors on $\mathsf{Set}$ to container functors on the wild category of types. Our main incentive for using \cubicalagda is that its path type restores the equivalence between bisimulation and coinductive equality. Thus, besides developing container theory in a more general setting, we also demonstrate the usefulness of \cubicalagda's path type to coinductive proofs.

%% file: sections/intro.tex
An inductive type is a type given by a list of constructors, each specifying a way to form an element of the type. Defining types inductively is a central idea in Martin-L\"of type theory (MLTT) \cite{martin-lof-1972}, with examples including the natural numbers, lists, and finite sets. In order for our inductive definitions to `make sense', meaning their induction principle can be properly expressed without leading to inconsistencies (see \cite[Section~5.6]{hott_book}),
we need to impose conditions on the general form of a constructor. The condition we would like to impose on our inductive definitions is that they are \textit{strictly positive}. Dual to inductive types is the notion of a coinductive type, i.e.\ a type defined by a list of destructors. While inductive types are described by the different ways we can construct them, coinductive types are described by the ways we can break them apart. Coinductive types are typically infinite structures, and examples include the conatural numbers and streams. As is the case for inductive definitions, coinductive definitions also ought to be strictly positive in order to avoid inconsistencies.\footnote{To be precise, in both cases we mean that the type's signature functor should be strictly positive.}
To ensure our type systems only admit such types, we seek a semantic description of strict positivity --- this is precisely what containers provide.

The theory of containers \cite{alt_fossacs_2003, alt_icalp_2004, abbott_tcs_2005, abbott_phd} (also referred to as polynomial functors in the literature \cite{poly_functors_hyland_gambino}) was developed to capture the concept of strictly positive data types in programming, and has been very useful in providing semantics for inductive types and inductive families. The original development of containers \cite{alt_fossacs_2003, alt_icalp_2004, abbott_tcs_2005, abbott_phd} uses a categorical language, where containers are presented as constructions in the internal language of locally cartesian closed categories (LCCCs) with disjoint coproducts and W-types (so-called Martin-L\"of categories), and a standard set-theoretic metatheory is used. Abbott et al.\ claim in \cite[Section 2.1]{abbott_tcs_2005} that these developments can also be interpreted as constructions in extensional MLTT. While Seely, Hofmann, and others show that this translation can be done in principle \cite{seely-1984, hofmann-1994}, the construction in type theory is not actually carried out.

In this paper, we remedy this by providing a formalisation of existing results by Abbott et al.~\cite{abbott_tcs_2005}. While the aforementioned paper implicitly assumes uniqueness of identity proofs (UIP) throughout, we do not assume this principle for any of the types involved in our formalisation. More precisely, the main contributions of this paper are:
\begin{itemize}
\item the formalisation of the results stating `container functors preserve initial algebras'  and `container functors preserve terminal coalgebras' (\cref{thm-least-fp,thm-greatest-fp}), and
\item  the generalisation of these results by proving them not in the category of sets, but in the wild category of types.
\end{itemize}\vspace{-6.5pt}

The formalisation is written in \cubicalagda~\cite{cubical_agda}, an extension of the \agda proof assistant which implements (a cubical~\cite{CCHM18} flavour of) homotopy type theory (HoTT). Although a lot of the interest around \cubicalagda is arguably due to its native support for the univalence axiom, our main motivation for using it is due to its treatment of equality as a path type, which restores the symmetry between inductive and coinductive reasoning in \agda. 
In intensional type theory, and therefore in vanilla \agda, working with inductive types is facilitated by using pattern matching and structural recursion. 
For coinductive types, while we do have copattern matching and some guarded corecursion, we cannot get very far when working in this setting. 
In particular, many equalities on coinductive types are impossible to prove in much the same way that equalities on function types cannot be proved: this requires some form of extensionality. Fortunately, \cubicalagda makes function extensionality provable, and many of the equalities on coinductive types that were previously impossible in vanilla \agda also become provable. With our formalisation, we contribute to the \agdaCubical library~\cite{The_Agda_Community_Cubical_Agda_Library_2025}, and hope to demonstrate that the practical developments brought to MLTT by cubical type theory extend beyond just computational univalence. 

Our formalisation is available
at \cite{code} in a fork of the \agdaCubical library. We have type-checked it using
versions $2.6.4.3$ and $2.6.5$ of \agda. We have taken
some liberties concerning the typesetting in this paper and,
although it should still be easy to follow, the syntax of the formalisation may
differ slightly from the paper.

The paper is organised as follows. In \cref{sec-background}, we give a brief overview of \agda and \cubicalagda, W- and M-types, and some container theory. In \cref{sec-settingup} we motivate and develop the constructions to be used in the main proofs, which are then given in \cref{sec-fixedpoints}. Lastly, we conclude in \cref{sec-conclusions} with some related work and future avenues of study.

%% file: sections/background.tex
In this section, we present some background material that aids in understanding the rest of the paper. We start by giving an overview of \agda and \cubicalagda concepts that will be used throughout the paper.  We then review \agdablue{W}- and \agdablue{M}-types, as well as the theory of containers adapted to wild categories.

\subsection{Agda and Cubical Agda}
\agda is a dependently typed functional programming language and proof assistant. It supports inductive data types, e.g.\ the unit type, the empty type, the natural numbers\\
\begin{minipage}[t]{0.32\textwidth}
\minipageskip
\ExecuteMetaData[agda/latex/Code.tex]{unit}
\end{minipage}
\begin{minipage}[t]{0.32\textwidth}
\minipageskip
\ExecuteMetaData[agda/latex/Code.tex]{empty}
\end{minipage}
\begin{minipage}[t]{0.33\textwidth}
\minipageskip
\ExecuteMetaData[agda/latex/Code.tex]{Nat}
\end{minipage}\\
and record types, e.g.\ $\Sigma$-types and (coinductive) streams.\\
\begin{minipage}[t]{0.55\textwidth}
\minipageskip
\ExecuteMetaData[agda/latex/Code.tex]{Sigma}
\end{minipage}
\begin{minipage}[t]{0.50\textwidth}
\minipageskip
\ExecuteMetaData[agda/latex/Code.tex]{stream}
\end{minipage}
\agda supports pattern matching on inductive data types, and dually supports copattern matching on record and coinductive data types, as shown below.\\
\begin{minipage}[t]{0.50\textwidth}
\minipageskip
\ExecuteMetaData[agda/latex/Code.tex]{even}
\end{minipage}
\begin{minipage}[t]{0.50\textwidth}
\minipageskip
\ExecuteMetaData[agda/latex/Code.tex]{from}
\end{minipage}\\
We note that \agda uses the syntax $(a : A) \to B\;a$ rather
than $\Pi_{a:A} B\;a$; we will use both notations interchangeably throughout the paper.
We also remark that the syntax $\{a : A\} \to B\;a$ is
available in \agda and denotes the same construction but with $a$ an
implicit argument.

%%
%% Main section begins here
%%
\cubicalagda \cite{cubical_agda} extends \agda with primitives from cubical type theory~\cite{CCHM18, hits-2018}. While the original motivation behind this extension was arguably to allow for native support for Voevodsky's univalence axiom~\cite{Voevodsky10cmu} and higher inductive types~\cite{LumsdaineShulman19}, we are primarily interested in \cubicalagda's representation of equality. The equality type in \cubicalagda, also called the \emph{path type}, restores the equivalence between bisimilarity and equality for
coinductive types. In order to explain how, let us briefly introduce
some of the elementary machinery of \cubicalagda.

The main novelty of \cubicalagda is the addition of an interval
(pre-)type \agdablue{I}. This type has two terms
$\agdagreen{i0},\agdagreen{i1}:\agdablue{I}$ denoting the endpoints of
the interval. It comes equipped with operations, such as
a `reversal' operation
$\agdablue{$\sim$} : \agdablue{I}\to\agdablue{I}$,
which allow us to internalise
the usual homotopical notion of a path and take that as our definition
of equality. Indeed, an equality $p$ between two points $x,y : A$,
denoted $p : x \;\agdablue{$\equiv$}\;y$, is a path in $A$ between $x$
and $y$, i.e.\ a function $p : \agdablue{I} \to A$ such that
$p \;\agdagreen{i0}$ is definitionally equal to $x$ and
$p \;\agdagreen{i1}$ is definitionally equal to $y$. To showcase some
elementary constructions of paths in \cubicalagda, consider e.g.\ the
constructions/proofs of reflexivity and symmetry.\\
\begin{minipage}[t]{0.45\textwidth}
\minipageskip
\ExecuteMetaData[agda/latex/Code.tex]{refl}
\end{minipage}
\begin{minipage}[t]{0.50\textwidth}
\minipageskip
\ExecuteMetaData[agda/latex/Code.tex]{sym}
\end{minipage}\\
\indent As in plain MLTT, there is also a notion of dependent path in \cubicalagda expressed by the primitive type called \agdablue{PathP}. This type generalises the path type by considering dependent functions $(i:\agdablue{I}) \to A\;i$ instead of only non-dependent ones $\agdablue{I} \to A$.
In this paper, we adopt an informal approach and write $a\; \agdablue{$\approxeq$} \;b$ for the statement that `$a:A$ is equal to $b:B$ modulo some path of types $p : A \;\agdablue{$\equiv$}\; B$'. The definition of the path of types will often be omitted from the main text, as it almost always can be automatically inferred from context. However, it will be included in a separate text box for the interested reader, although it can generally be ignored by the reader focusing on the broader picture.

On a related note, we also mention here that the \cubicalagda primitives allow us to define $\agdablue{transport} : \agdaeq{A}{B} \to A \to B$. As usual we have that, for any $p : \agdaeq{A}{B}$, the type of dependent paths $a\;\agdablue{$\approxeq$}\;b$ modulo $p$ is equivalent to the type $\agdaeq{\agdablue{transport}\;p\;a}{b}$.

One of the key advantages of \cubicalagda's treatment of equality is
that it renders function extensionality a triviality:
\ExecuteMetaData[agda/latex/Code.tex]{funext}
A consequence of this is that we can use the types $\agdaeq{f}{g}$ and
$(x : A) \to \agdaeq{f\;x}{g\;x}$ interchangeably without having to
worry about introducing any bureaucracy when moving from one to the
other; in \cubicalagda, equality of functions \emph{is by definition}
pointwise equality.  In a similar way, and especially important for
us, the equality type of a coinductive type \emph{is}
bisimulation, modulo copattern matching. We can define \agdablue{id} below 
by first introducing the path variable $i$, after which we are required to construct 
an element of $\agdablue{Stream}\,A$ matching the endpoints $xs$ at $i=0$ and $ys$ at $i=1$,
which we do using copattern matching. 
In particular, we can show that \agdablue{id} is an equivalence, which then allows us to prove \agdablue{path≃bisim} stated below \cite{cubical_agda}, where \agdablue{≃} denotes equivalence (and if we assume univalence, which we do not require for this paper, we can even show that $(xs\, \agdablue{$\equiv$}\, ys)\, \agdablue{$\equiv$}\, (xs\, \agdablue{$\approx$}\, ys)$). To the best of our knowledge, this feature is unique to \cubicalagda.\\
\begin{minipage}[t]{0.45\textwidth}
\minipageskip
\ExecuteMetaData[agda/latex/Code.tex]{bisim}
\end{minipage}
\begin{minipage}[t]{0.50\textwidth}
\minipageskip
\ExecuteMetaData[agda/latex/Code.tex]{streamid}
\ExecuteMetaData[agda/latex/Code.tex]{pathbisim}
\end{minipage}

\vspace{-5pt}
\subsection{The W-type and the M-type} \label{sec:w-m}

The type \agdablue{W}, due to Martin-L\"of \cite{martin-lof-1982, martin-lof-1984}, is the type of well-founded, labelled trees. A tree of type \agdablue{W} can be infinitely branching, but every path in the tree is finite. \agdablue{W} takes two parameters $S:\agdablue{Type}$ and $P\colon S\to\agdablue{Type}$. We think of $S$ as the type of shapes of the tree, and for a given shape $s:S$, the tree has $(P\; s)$-many positions. The key property of \agdablue{W} is that it is the universal type for strictly positive inductive types, i.e.\ any strictly positive inductive type can be expressed using \agdablue{W}. 

\ExecuteMetaData[agda/latex/Code.tex]{W}

\vspace{-5pt}
\begin{example}\label{example1}
	We encode the type of natural numbers $\agdablue{$\mathbb{N}$}$ by defining \agdablue{S} and \agdablue{P} as below (where $A\;\agdablue{$\uplus$}\; B$ is the sum type of $A$ and $B$ with constructors \agdagreen{inl} and \agdagreen{inr}).
	 \begin{align*}
	      \agdablue{S} &= \agdablue{$\top$}\; \agdablue{$\uplus$}\; \agdablue{$\top$} \\
	      \agdablue{P}\, (\agdagreen{inl}\; \_) &= \agdablue{$\bot$}\\
	      \agdablue{P}\, (\agdagreen{inr}\; \_) &= \agdablue{$\top$}
	 \end{align*} 
	\agdablue{S} encodes the possible constructors we can choose (\agdagreen{inl} is for \agdagreen{zero}, \agdagreen{inr} is for \agdagreen{succ}), and \agdablue{P} encodes the number of subtrees (or recursive arguments) each choice of \agdablue{S} has. Thus $\agdablue{W S P}$ encodes $\agdablue{$\mathbb{N}$}$. For example, \agdagreen{zero} and \agdagreen{succ zero}\ 
	are represented by the trees below respectively.
	
	\begin{minipage}[t]{0.45\linewidth}
		\adjustbox{scale=0.8,center}{
			\begin{tikzcd}
				& \agdagreen{inl tt}
			\end{tikzcd}
		}
	\end{minipage}
	\begin{minipage}[t]{0.45\linewidth}
		\adjustbox{scale=0.8,center}{
			\begin{tikzcd}
				& \agdagreen{inr tt}\arrow{dr}{} &\\
				& & \agdagreen{inl tt}
			\end{tikzcd}
		}
	\end{minipage}
\end{example}

The type \agdablue{M}, first studied by Abbott et al.\ \cite{abbott_tcs_2005} and van den Berg and De Marchi \cite{m-types}, is the type of non-well-founded, labelled trees. A tree of type \agdablue{M} can have both finite and infinite paths. \agdablue{M} takes two parameters $S:\agdablue{Type}$ and $P\colon S\to\agdablue{Type}$, and we think of them similarly as for \agdablue{W}. Dually to \agdablue{W}, \agdablue{M} is the universal type for strictly positive \emph{coinductive} types.

\ExecuteMetaData[agda/latex/Code.tex]{M}

\begin{example} \label{ex:conats}
	If we define \agdablue{S} and \agdablue{P} as in \cref{example1}, then $\agdablue{M S P}$ is an encoding of the conatural numbers $\agdablue{$\mathbb{N}\infty$}$. 
\ExecuteMetaData[agda/latex/Code.tex]{Ninfty}
	Apart from having all the natural numbers as its elements, $\agdablue{$\mathbb{N}\infty$}$ also has an `infinite number' whose predecessor is itself. This number 
	is represented by the infinite tree shown below. This $\agdablue{M}$-tree clearly has an infinite path.\vspace{5pt}
	
	\begin{minipage}[t]{0.45\linewidth}
		\adjustbox{scale=0.8,center}{
			\begin{tikzcd}
				& \agdagreen{inr tt}\arrow[loop right]
			\end{tikzcd}
		}
	\end{minipage}
\end{example}

We will see in \cref{sec-fixedpoints} that it is useful to have an
explicit account of the coinduction principle of \agdablue{M}, which
states that any two $m_0 , m_1 : \agdablue{M} \,S\,Q$ that can be
related by a bisimulation are equal. 
In \cubicalagda, we can define what it means for a relation $R$ on \agdablue{M} to be a bisimulation using \agdablue{M-R} below --- $R$ has to relate two elements $m_0$ and $m_1$ of \agdablue{M} whenever their \agdapink{shape}s are equal and their \agdapink{pos}itions are related by $R$.
\ExecuteMetaData[agda/latex/Code.tex]{M-R}
\ \vspace{-.7cm} \\
\mycomment{Above, $q$-$eq$ is a dependent path over the path of types $(λ\; i \;→\; Q \;(\agdapink{s-eq}\, i))$}\\
\ \vspace{-.3cm} \\
We can then prove the coinduction principle using interval abstraction and
copattern matching.
\ExecuteMetaData[agda/latex/Code.tex]{MCoind}
Above, $\AgdaFunction{$q_0$}$, $\AgdaFunction{$q_1$}$, and
$\AgdaFunction{$q_2$}$ are all of the form
$\AgdaFunction{transport} \,\AgdaFunction{...}\, q$. The constructions
of these transports use some rather technical cube algebra; we omit
the details and refer the interested reader to the formalisation.

\subsection{Containers}
\label{subsec-containers}

The container and container functor definitions in this section are adapted from \cite{abbott_tcs_2005} to the wild category of types \agdablue{Type}.\footnote{We note that we use \agdablue{Type} to refer to both the wild category of types as well as \cubicalagda's universe of types.} A wild category is a precategory as in the HoTT book \cite{hott_book}, except the type of morphisms need not be an h-set (i.e.\ it need not satisfy UIP); this is also called a precoherent higher category in \cite{josh-chen-2025}. We observe that \agdablue{Type} has triangle and pentagon coherators, and is therefore a 2-coherent category in the sense of  \cite{josh-chen-2025}, and moreover its coherators are trivial \cite[Example 2.2.5]{josh-chen-2025}.
The definitions of container functor algebras are standard category theory definitions also adapted to this setting, with the initial algebra definition being motivated by \cite[Definition 5]{kraus-raumer-2021}.

\begin{definition}\label{def:container}
 	A \underline{(unary) container} is given by a pair of types $S : \agdablue{Type}$ and $P : S \to\agdablue{Type}$, which we write as $\cont{S}{P}$. 
\end{definition}

In practice, many data types are parameterised by one or more types. For example, $\agdablue{List}\;(A~:~\agdablue{Type}):\agdablue{Type}$ and $\agdablue{Vec}\;(A:\agdablue{Type}) \colon \agdablue{$\mathbb{N}$}\to\agdablue{Type}$ are both parameterised by the type $A$ of data to be stored in them. In order to be able to reason about such parameterised data types, as well as to construct fixed points of containers, we will need containers parameterised by some (potentially infinite) indexing type $I$. We call these $I$-ary containers.

\begin{definition}
 	An \underline{$I$-ary container} is given by a pair $S : \agdablue{Type}$ and $\mathbf{P} : I\to S \to \agdablue{Type}$, which we write as $\cont{S}{\mathbf{P}}$.
\end{definition}
Above, $\mathbf{P}$ can be thought of as $I$ families of families over $S$. We will sometimes write $P_0, P_1, \dots$ instead of $\mathbf{P}\; i_0, \mathbf{P}\; i_1, \dots$ to enumerate the families over $S$. 

\begin{example}\label{ex:list-N-Fin}
	Given a type $A$, the (unary) container representation of $\agdablue{List}\; A : \agdablue{Type}$ is given by $\cont{\agdablue{$\mathbb{N}$}}{\agdablue{Fin}}$.\footnote{The precise meaning of this is that $\contfuncX{\agdablue{$\mathbb{N}$}}{\agdablue{Fin}}{A} \cong \mu X. F_ {\agdablue{List}}(A, X)$, see \cref{example:list-param}.} The shape of a list is a natural number $n$ representing its length, and there are $n$-many positions for data to be stored in a list, represented by $\agdablue{Fin}\; n$, the type of finite sets of size $n$.
\end{example}
Unary containers are trivially $I$-ary containers when $I=\agdablue{$\top$}$, so henceforth in this section we will only consider $I$-ary containers.

To every container $\cont{S}{\mathbf{P}}$, we associate a wild functor which maps a family of types
$\mathbf{X} : I \to\agdablue{Type}$ to a choice of shape $s : S$, and for every $i:I$ and position $\mathbf{P}\; i\; s$ associated to $s$, a value of type $\mathbf{X}\; i$ to be stored at that position.

\begin{definition}\label{defn:cont-func}
 	The \underline{container functor} associated to an $I$-ary container $\cont{S}{\mathbf{P}}$ is the wild functor $\contfunc{S}{\mathbf{P}}\colon(I \to \agdablue{Type}) \to \agdablue{Type}$ with the following actions on objects and morphisms.\footnote{Here, $I\to\agdablue{Type}$ refers to the wild category of $I$-indexed families of types.}
 	\begin{itemize}
 		\item Given $\mathbf{X} : I\to\agdablue{Type}$, we define $\contfuncX{S}{\mathbf{P}}{\mathbf{X}}\coloneqq\underset{s:S}\nsum \pa[\Big]{\underset{i:I}\nprod\, \mathbf{P}\, i\, s \to \mathbf{X}\, i}$.
 		\item Given $\mathbf{X}, \mathbf{Y} : I\to\agdablue{Type}$, and a morphism $f\colon \underset{i:I}\nprod\  \mathbf{X}\, i \to \mathbf{Y}\, i$, we define
 		\[\contfuncX{S}{\mathbf{P}}{f}\,(s , g)\coloneqq (s , f\circ g)\]
 		%\[\contfuncX{S}{\mathbf{P}}{f}\,(s , g)\coloneqq (s , \lambda\, i\, p.\, f\, i\, (g\, i\, p))\]
 		for $s : S$ and $g\colon \underset{i:I}\nprod\ \mathbf{P}\, i\, s \to \mathbf{X}\, i$, where $f \circ g$ is composition in $I\to\agdablue{Type}$. \qedhere
 	\end{itemize}
\end{definition}

As a special case of \cref{defn:cont-func}, given an $(I+1)$-ary container $F=\cont{S}{\mathbf{R}}$, we will later need to write it in a way where we single out one component from it. We split $\mathbf{R}$ into $\mathbf{P}$ and $Q$ and write $F$ as having components $S : \agdablue{Type}, \mathbf{P}\colon I\to S\to \agdablue{Type}, Q\colon S\to\agdablue{Type}$, and use the notation $F = (\cont{S}{\mathbf{P}},Q)$. Given $\mathbf{X}\colon I\to \agdablue{Type}$ and $Y : \agdablue{Type}$, then $S, \mathbf{P},$ and $Q$ satisfy the below.
\begin{equation}\label{eqn:cont-func}
	\llbracket \cont{S}{\mathbf{P}},Q\rrbracket (\mathbf{X}, Y) = \sum_{s:S} \pa[\Big]{(i : I)\to  \mathbf{P}\, i\, s\to \mathbf{X}\, i} \times (Q\, s\to Y) 
\end{equation}

\begin{example}\label{example:list-param}
	The signature functor of \agdablue{List}, $F_{\agdablue{List}}\, (A, X) = \agdablue{$\top$} \uplus (A \times X)$,
	is isomorphic to the container functor $\contfuncX{\agdablue{S}}{\agdablue{P$_0$}, \agdablue{P$_1$}}{(A,X)}$ via the below definitions.
	
	\begin{minipage}[t]{0.32\linewidth}
		\begin{align*}
			\agdablue{S} &: \agdablue{Type}\\
			\agdablue{S} &\coloneqq \agdablue{$\top$}\; \agdablue{$\uplus$}\; \agdablue{$\top$}
		\end{align*}
	\end{minipage}
	\begin{minipage}[t]{0.32\linewidth}
		\begin{align*}
			&\agdablue{P$_0$}\colon \agdablue{S}\to \agdablue{Type}\\
			&\agdablue{P$_0$}\; (\agdagreen{inl tt}) \coloneqq \agdablue{$\bot$}\\
			&\agdablue{P$_0$}\; (\agdagreen{inr tt}) \coloneqq \agdablue{$\top$}
		\end{align*}
	\end{minipage}
	\begin{minipage}[t]{0.32\linewidth}
		\begin{align*}
			&\agdablue{P$_1$}\colon \agdablue{S}\to \agdablue{Type}\\
			&\agdablue{P$_1$}\; (\agdagreen{inl tt}) \coloneqq \agdablue{$\bot$}\\
			&\agdablue{P$_1$}\; (\agdagreen{inr tt}) \coloneqq \agdablue{$\top$}
		\end{align*}
	\end{minipage} \vspace{5pt}
	
	The type of shapes \agdablue{S} reflects the fact that there are two ways to construct a list, i.e.\ as either $\mathsf{nil}$ or $\mathsf{cons}$. \agdablue {P$_0$} defines positions for the parameter $A$ and \agdablue{P$_1$} defines positions for the recursive argument $X$. \cref{ex:list-N-Fin} corresponds to taking the least fixed point of this functor with respect to $X$.
\end{example}
\begin{example}
	Container functors allow us to view strictly positive types simply as memory locations in which data can be stored. The container functor associated to $\cont{\agdablue{$\mathbb{N}$}}{\agdablue{Fin}}$ allows us to represent concrete lists. The list of $\mathsf{Char}$s [`r', `e', `d'] is represented as $(3, (0~\mapsto~\text{`$r$'}; 1~\mapsto~\text{`$e$'};\\ 2~\mapsto~\text{`$d$'})) : \nsum_{n:\agdablue{$\mathbb{N}$}} (\agdablue{Fin}\,n \to\mathsf{Char})$. 
\end{example}

The two main results formalised in this paper concern the initial algebra and terminal coalgebra of a container functor. We define explicitly what we mean by these in the setting of wild categories and wild functors. We note that for (at least a naive definition of) a functor of wild categories $F\colon \cat{C}\to\cat{C}$, it is not in general the case that $F$-algebras form a wild category. However, this is the case for container functors. This follows from the properties of \agdablue{Type} we mentioned at the start of the section.

\begin{definition}\label{def:algebras}
	For a (unary) container functor $\llbracket F \rrbracket\colon \agdablue{Type}\to\agdablue{Type}$, the \underline{wild category of} \underline{$\llbracket F\rrbracket$-algebras}, denoted $\mathsf{Alg}_{\llbracket F \rrbracket}$, is defined as follows.
	\begin{itemize}
		\item Objects are \underline{algebras}: pairs $(X\colon\agdablue{Type}, \alpha\colon \llbracket F \rrbracket\, X\to X)$.\footnote{$X$ is sometimes called the carrier.}
		\item A morphism of algebras $(X, \alpha) \to (Y, \beta)$ is a function $f\colon X\to Y$ such that the following square commutes.
		\[
		\begin{tikzcd}
			\llbracket F \rrbracket\, X\arrow{r}{\llbracket F\rrbracket\; f}\arrow[swap]{d}{\alpha} & \llbracket F \rrbracket\, Y\arrow{d}{\beta}\\
			X\arrow[swap]{r}{f} & Y
		\end{tikzcd}
		\]
	\end{itemize}
\end{definition}

\begin{definition}
	The \underline{initial algebra} of a (unary) container functor $\llbracket F \rrbracket\colon \agdablue{Type}\to\agdablue{Type}$ is an algebra $(I, \iota)$ such that for every other algebra $(X, \alpha)$, $\mathsf{Alg}_{\llbracket F \rrbracket} ((I, \iota), (X, \alpha))$ is contractible.
\end{definition}

The \underline{wild category of $\llbracket F \rrbracket$-coalgebras} $\mathsf{Coalg_{\llbracket F \rrbracket}}$ is dual to \cref{def:algebras}, where the objects, \underline{coalgebras}, are defined as pairs $(X\colon\agdablue{Type}, \alpha\colon X\to \llbracket F \rrbracket\, X)$. The terminal coalgebra is then an object $(T, \tau)$ such that for every other coalgebra $(X, \alpha), \mathsf{Coalg}_{\llbracket F \rrbracket} ((X, \alpha), (T, \tau))$ is contractible.

%% file: sections/settingup.tex
In this section, we state precisely what it is that we want to prove and start attacking the problem. We construct a candidate initial algebra and terminal coalgebra for a general container functor, which in the following section we prove to be correct. We also discuss a generalised induction principle for the inductive family \agdablue{Pos} of finite paths in a tree.

\subsection{Calculation of the initial algebra and terminal coalgebra}

Given a container functor $\llbracket F\rrbracket \colon (I + 1 \to \agdablue{Type})\to \agdablue{Type}$, which we write as $F = (\cont{S}{\mathbf{P}}, Q)$, we need to specify container functors
\begin{align*}
\contfunc{A_\mu}{\mathbf{B}_\mu}&\colon (I \to\agdablue{Type})\to\agdablue{Type}\\
\contfunc{A_\nu}{\mathbf{B}_\nu}&\colon (I \to\agdablue{Type})\to\agdablue{Type}
\end{align*}
such that
\begin{align*}
\contfuncX{A_\mu}{\mathbf{B_\mu}}{\mathbf{X}} &\cong \mu Y.\llbracket F\rrbracket (\mathbf{X},Y),\\
\contfuncX{A_\nu}{\mathbf{B_\nu}}{\mathbf{X}} &\cong \nu Y.\llbracket F\rrbracket (\mathbf{X},Y).
\end{align*}
Above, and for the remainder of the paper, $\cong$ is used to denote an equivalence of types.\footnote{In HoTT, there are several different notions of type equivalence. In our formalisation, we primarily use a definition in terms of \emph{quasi-inverses}~\cite[Definition 2.4.6]{hott_book}, i.e.\ a function with an explicit inverse. All statements in this paper are independent of this particular choice and can be read with any other reasonable notion of equivalence in mind.} The symbols
$\mu$ and $\nu$ denote partial operators taking a wild functor to the carrier of its initial algebra or terminal coalgebra respectively, if they exist. The notation $\mu Y.\llbracket F\rrbracket (\mathbf{X}, Y)$ is shorthand for (the carrier of) the initial algebra of the wild functor $G$ defined by $G\, Y \coloneqq \llbracket F\rrbracket (\mathbf{X}, Y)$, and similarly for $\nu$. 

We now illustrate how we calculate containers $(\cont{A_\mu}{\mathbf{B_\mu}})$ and $(\cont{A_\nu}{\mathbf{B_\nu}})$ in $I$ parameters to make the above isomorphisms hold.
Calculating $A_\mu$ and $A_\nu$ is straightforward. If we set $\mathbf{X} = \agdablue{$\top$}$ in the above, we get
\begin{align*}
A_\mu &\cong \contfuncX{A_\mu}{\mathbf{B}_\mu}{\agdablue{$\top$}} \cong \mu Y. \llbracket F\rrbracket(\agdablue{$\top$},Y) \cong \mu Y. \sum_{s:S} (Q\; s\to Y) \cong \mu Y. \contfuncX{S}{Q}{Y} \cong \agdablue{W}\,S\,Q\\
A_\nu &\cong \contfuncX{A_\nu}{\mathbf{B}_\nu}{\agdablue{$\top$}}\cong \nu Y. \llbracket F\rrbracket(\agdablue{$\top$},Y) \cong \nu Y. \sum_{s:S} (Q\; s\to Y) \cong \nu Y. \contfuncX{S}{Q}{Y} \cong \agdablue{M}\,S\,Q.
\end{align*}  
The last step follows from the fact that the least (resp.\ greatest) fixed point of the container functor in one variable $\contfunc{S}{Q}$ is $\agdablue{W}\,S\,Q$ ($\agdablue{M}\,S\,Q$), the \agdablue{W}-type (\agdablue{M}-type) with shapes $S$ and positions $Q$.

Calculating $\mathbf{B_\mu}\colon I\to \agdablue{W}\,S\,Q\to\agdablue{Type}$ and $\mathbf{B_\nu}\colon I\to \agdablue{M}\,S\,Q\to\agdablue{Type}$ is a bit more involved. Our reasoning below applies to both $\agdablue{W}\,S\, Q$ and $\agdablue{M}\,S\,Q$, so we consider any fixed point $\phi$ of the container functor $\contfunc{S}{Q}$ and construct $\mathbf{B}\colon I\to \phi\to \agdablue{Type}$. Being a fixed point of $\contfunc{S}{Q}$ means that $\phi$ consists of a carrier $\agdapink{C}:\agdablue{Type}$ together with an isomorphism, $\agdapink{$\chi$}\colon \contfuncX{S}{Q}{\agdapink{C}}\cong \agdapink{C}$.
\ExecuteMetaData[agda/latex/Code.tex]{FixedPoint}
In particular, we have \agdablue{WAlg : FixedPoint} whose carrier is $\agdablue{W}\,S\,Q$ and \agdablue{MAlg : FixedPoint} whose carrier is $\agdablue{M}\,S\,Q$ (for the \agdapink{$\chi$} components, we refer the reader to the formalisation).

If $\contfunc{\agdapink{C}}{\mathbf{B}}{\mathbf{X}}$ is to be a fixed point of $\llbracket F \rrbracket (\mathbf{X}, -)$, by Lambek's theorem \cite{lambek_1968}, the following isomorphism is induced.
\begin{equation}
\llbracket F\rrbracket (\mathbf{X}, \contfuncX{\agdapink{C}}{\mathbf{B}}{\mathbf{X}})\cong \contfuncX{\agdapink{C}}{\mathbf{B}}{\mathbf{X}}
\label{iso}
\end{equation}  
By massaging the left hand side of this isomorphism, we can write it as a container functor in terms of only $\mathbf{X}$.
\begin{align*}
  & \phantom{=} \sum_{s:S} \pa[\Big]{\pa[\Big]{\prod_{i}(\mathbf{P}\; i\; s\to \mathbf{X}\; i)} \times (Q\; s\to\contfuncX{\agdapink{C}}{\mathbf{B}}{\mathbf{X})}} && 
  \begin{tabular}{p{2.5cm}}
  	\small
  	\hspace{-.45cm}using \cref{eqn:cont-func}
  \end{tabular}\\
  & = \sum_{s:S} \pa[\Big]{\pa[\Big]{\prod_{i}(\mathbf{P}\; i\; s\to \mathbf{X}\; i)} \times \pa[\Big]{Q\; s\to \sum_{c:\agdapink{C}} \pa[\Big]{\prod_{i} (\mathbf{B}\; i\; c\to \mathbf{X}\; i)}}} &&
    \begin{tabular}{p{2.5cm}}
      \hspace{-.45cm}definition of $\danascott$
    \end{tabular}\\
    & \begin{aligned}
        \cong \sum_{s:S} \pa[\Big]{&\pa[\Big]{\prod_{i}(\mathbf{P}\; i\; s\to \mathbf{X}\; i)} \times
          {\sum_{f\colon Q\; s\to \agdapink{C}} \pa[\Big]{\prod_{q:Q\; s}\prod_{i} (\mathbf{B}\; i\; (f\; q) \to \mathbf{X}\; i)}}}
        \end{aligned} &&
    \begin{tabular}{p{2.5cm}}
      \hspace{-.45cm}distributivity of $\Pi$\\
      \hspace{-.45cm}over $\Sigma$
    \end{tabular}\\
   & \cong \sum_{(s,f) \colon  \sum_{s:S} (Q\; s\to \agdapink{C})} \pa[\Big]{\prod_{i} \pa[\Big]{\mathbf{P}\; i\; s + \sum_{q : Q\; s} (\mathbf{B}\; i\; (f\; q))} \to \mathbf{X}\; i} &&
     \begin{tabular}{p{2.5cm}}
       \hspace{-.75cm}commutativity of $\times$\\
       \hspace{-.75cm}and $(A\, \agdablue{$\uplus$}\, B) \to C\,\cong$ \\
       \hspace{-.75cm}$(A \to C) \times (B \to C)$
     \end{tabular}\\
  & = \bigcontfuncX{\sum_{s:S} (f\colon Q\; s\to \agdapink{C})}{\pa[\Big]{\lambda\; i.\; \mathbf{P}\; i\; s + \sum_{q : Q\; s} (\mathbf{B}\; i\; (f\; q))}}{\mathbf{X}} &&
     \begin{tabular}[t]{p{2.5cm}}
       \hspace{-.75cm}definition of $\danascott$
     \end{tabular}
\end{align*}

The induced isomorphism \eqref{iso} can then be written as
\[\bigcontfuncX{\sum_{s:S} (f\colon Q\; s\to \agdapink{C})}{\pa[\Big]{\lambda\; i.\; \mathbf{P}\; i\; s + \sum_{q : Q\; s} \mathbf{B}\; i\; (f\; q)}}{\mathbf{X}} \cong \contfuncX{\agdapink{C}}{\mathbf{B}}{\mathbf{X}}.\]

We already have the isomorphism $\agdapink{$\chi$}\colon \sum_{s:S} (f\colon Q\; s\to \agdapink{C})\cong \agdapink{C}$ on shapes.  We will also need the below isomorphism on positions for $i:I$ and $c:\agdapink{C}$. 
We call \agdablue{$\xi$} the map in the inverse direction of $\agdapink{$\chi$}$ and use the notation ($\phi\ \agdablue{$\xi_0$}$) and ($\phi\  \agdablue{$\xi_1$})$ for its first and second projections, so that  $(\phi\ \agdablue{$\xi_0$})\; c: S$ and $(\phi\ \agdablue{$\xi_1$})\; c : Q\; ((\phi\ \agdablue{$\xi_0$})\; c) \to \agdapink{C}$.
\[\pa[\Big]{\mathbf{P}\; i\; ((\phi\; \agdablue{$\xi_0$})\; c) + \sum_{q:Q\; ((\phi\; \agdablue{$\xi_0$})\; c)} \mathbf{B}\; i\; ((\phi\; \agdablue{${\xi}_1$})\; c\; q)} \cong \mathbf{B}\; i\; c\]
We use this as our definition of $\mathbf{B}$, which we hereafter call \agdablue{Pos}, as an inductive family over \agdapink{C}. In our code, \agdablue{Pos} is also parameterised by a fixed point $\phi$.
\ExecuteMetaData[agda/latex/Code.tex]{Pos}
It turns out that \agdablue{Pos} works for both cases: we set $\mathbf{B}_\mu = \agdablue{Pos WAlg}$ and $\mathbf{B}_\nu = \agdablue{Pos MAlg}$. It is not immediately clear that choosing $\mathbf{B}_\nu$ to be an inductive (and not coinductive) family over $\agdablue{M}\,S\,Q$ would be the right choice in the coinductive case, so we explain why this works in more detail. Intuitively, we can think of \agdablue{Pos} as the type of finite paths through a \agdablue{W} or \agdablue{M} tree.
To see this more clearly, we look at \agdablue{Pos} specified to $\phi = \agdablue {MAlg}, I = \agdablue{$\top$},$ and $\mathbf{P}\ \agdagreen{tt}\ s \coloneqq \agdablue{$\top$}$ (which implies we only consider unary containers), which would be equivalent to \agdablue{PosM} below.
\ExecuteMetaData[agda/latex/Code.tex]{PosM}
Now, as an example, recall from \cref{sec:w-m} the \agdablue{M} trees encoding 0, 1, and $\infty$ respectively of type $\agdablue{$\mathbb{N}\infty$}$.

\begin{minipage}[t]{0.32\linewidth}
		\adjustbox{scale=0.8,center}{
				\begin{tikzcd}
						& \agdagreen{inl tt}
					\end{tikzcd}
			}
	\end{minipage}
\begin{minipage}[t]{0.32\linewidth}
		\adjustbox{scale=0.8,center}{
				\begin{tikzcd}
						& \agdagreen{inr tt}\arrow{dr}{} &\\
						& & \agdagreen{inl tt}
					\end{tikzcd}
			}
	\end{minipage}
\begin{minipage}[t]{0.32\linewidth}
		\adjustbox{scale=0.8,center}{
				\begin{tikzcd}
						& \agdagreen{inr tt}\arrow[loop right]
					\end{tikzcd}
			}
	\end{minipage}
Now we look at the elements of \agdablue{PosM}\ (\agdablue{M S P}), where recall $\agdablue{M S P}\cong \agdablue{$\mathbb{N}\infty$}$, given the different elements $0, 1,$ and $\infty$ of $\agdablue{$\mathbb{N}\infty$}$. For the first tree (encoding 0), \agdablue{PosM} would consist solely of the element \agdagreen{here}, because we cannot construct anything via \agdagreen{below}, since $\agdablue{P}\;(\agdagreen{inl tt})$ is empty. 
For the second tree (encoding 1), \agdablue{PosM} consists of \agdagreen{here} and \agdagreen{below here}, since $\agdablue{P}\;(\agdagreen{inr tt})$ is now \agdablue{$\top$}. 
For the third tree (encoding $\infty$), \agdablue{PosM} consists of \agdagreen{here}, \agdagreen{below here}, \agdagreen{below} (\agdagreen{below here}), and so on, ad infinitum. Although \agdablue{M} trees can have infinite paths, like in the third case, any position (i.e.\ where data is stored in the tree, even though this example does not involve payloads) is obtained via a finite path, and since \agdablue{PosM} encodes exactly the finite paths, it is precisely what is required. We verify this is actually the case in \cref{sec-fixedpoints}. 

\subsection{Generalised induction principle for \agdablue{Pos}}

We take the opportunity to mention the induction
principle for \agdablue{Pos}, which will come in useful later. In general, given a fixed point $\phi$, an
index $i : I$, and a family of types $A : (c
: \phi \;.\agdapink{C}) \to \agdablue{Pos}\;\phi\;i\;c \to \agdablue{Type}$
equipped with
\begin{itemize}
\item $h : \{c : \phi \;.\agdapink{C}\} \;(p : P \;i\;((\phi\ \agdablue{$\xi_0$})\; c)) \to A\;c\; (\agdagreen{here}\;p)$
\item $b : \{c : \phi \;.\agdapink{C}\} \;(q : Q \;((\phi\ \agdablue{$\xi_0$})\; c))\; (p : \agdablue{Pos}\;\phi\;i\;((\phi\ \agdablue{$\xi_1$})\; c\;q)) \to A\;((\phi\ \agdablue{$\xi_1$})\; c\;q) \;p\; \to\\ \phantom{b : } A \;c\; (\agdagreen{below}\;q\;p)$
\end{itemize}
induces, in the obvious way, a dependent function $(c : \phi \;.\agdapink{C})\;
(p : \agdablue{Pos}\;\phi\;i\;c ) \to A \;c\;p$. In \cubicalagda, this
is precisely the induction principle we get from performing a
standard pattern matching.
In practice, however, this induction principle is quite limited. The
primary difficulty we run into is in the case where $A$ is only defined over
$(d:D)$ and $\agdablue{Pos}\;\phi\;i\;(f\;d)$ for some fixed function
$f : D \to \phi \;.\agdapink{C}$. In this case, the induction
principle above does not apply since $A$ is not defined over all of
$\phi\;.\agdapink{C}$ (this is entirely analogous to how path
induction does not apply to paths with fixed endpoints). There are, of
course, special cases when the induction principle is still
applicable: for instance, when $f$ is a retraction. In fact, we only
need $f$ to satisfy a weaker property, namely the following.
\vspace{.3cm}
\\
\begin{minipage}{0.5\linewidth}
\begin{definition}\label{def-pos-ind}
Given a fixed point $\phi$, a function $f :
D \to \phi \;.\agdapink{C}$ is called a \emph{$\phi$-retraction}
if for any $d:D$, the lift $\widehat{f}_d$ in the diagram to the right
exists.
\end{definition}
\end{minipage}
\begin{minipage}[r]{0.5\linewidth}
\vspace{-.6cm}
\[
\begin{tikzcd}[column sep = 4em, ampersand replacement=\&]
D \\
Q \;((\phi\; \agdablue{$\xi_0$})\; (f \;d)) \&\& \agdapink{C}
\arrow["f", from=1-1, to=2-3]
\arrow["{\widehat{f}_d}", dotted, from=2-1, to=1-1]
\arrow["(\phi\; \agdablue{$\xi_1$})\; (f\;d)"', from=2-1, to=2-3]
\end{tikzcd}
\]
\end{minipage}
\begin{lemma}[Generalised \agdablue{Pos} induction]\label{lem-pos-ind}
Let $\phi$ be a fixed point, $i:I$ an index, and $f :
D \to \phi \;.\agdapink{C}$ a $\phi$-retraction. Let $A : (d :
D) \to \agdablue{Pos}\;\phi\;i\;(f\;d) \to \agdablue{Type}$ be a
dependent type equipped with
\begin{itemize}
\item $h : \{d : D\} \;(p : P \;i\;((\phi\; \agdablue{$\xi_0$})\; (f\;d)))\to A\;d\; (\agdagreen{here}\;p)$
\item $b : \{d : D\} \;(q : Q \;((\phi\; \agdablue{$\xi_0$})\; (f\;d)))\; (p : \agdablue{Pos}\;\phi\;i\;((\phi\; \agdablue{$\xi_1$})\; (f\; d)\;q)) \to  A \;(\widehat{f}_d\;q) \;\widehat{p}\; \to \\ \phantom{b : } A \;d\; (\agdagreen{below}\;q\;p)$
\end{itemize}
where $\widehat{p}$ is $p$ transported along the witness of the fact
the diagram in \cref{def-pos-ind} commutes. This data induces a dependent function $(d : D) \;(p : \agdablue{Pos}\;\phi\;i\;(f\;d)) \to A\;d\;p$.

\end{lemma}
\begin{proof}[Proof sketch]
The induction principle follows immediately from the usual induction principle for \agdablue{Pos} but with the family $\widehat{A} : (c : \phi\;.\agdapink{C}) \to \agdablue{Pos}\;\phi\;i\;c \to \agdablue{Type}$ defined by
\[\widehat{A}\;c\;p \coloneqq (d : D) (t : \agdaeq{c}{f\;d}) \to A \; d \; \widehat{p} \]
where $\widehat{p}$ denotes the result of transporting $p$ along $t : \agdaeq{c}{f\;d}$. 
We obtain the appropriate statement by setting $c := f \;d$.
\end{proof}

%% file: sections/fixedpoints.tex
Let us now show that the constructions from \cref{sec-settingup} are
correct: $\contfuncX{\agdablue{W}\,S\,Q}{\agdablue{Pos
WAlg}}{\mathbf{X}}$ is the initial $\llbracket F \rrbracket
(\mathbf{X}, -)$-algebra carrier, and
$\contfuncX{\agdablue{M}\,S\,Q}{\agdablue{Pos MAlg}}{\mathbf{X}}$ is
the terminal $\llbracket F \rrbracket (\mathbf{X}, -)$-coalgebra carrier. The
proofs in this section mostly follow those given
in \cite{abbott_tcs_2005}, but in the more general (UIP-free) setting of \agdablue{Type} instead of $\mathsf{Set}$. 

 We start off by showing that $\contfuncX{\agdablue{W}\;S\;Q}{\agdablue{Pos WAlg}}{\mathbf{X}}$ is the initial $\llbracket \cont{S}{\mathbf{P}}, Q \rrbracket (\mathbf{X}, -)$-algebra carrier. This proof is relatively straightforward.

\begin{theorem}\label{thm-least-fp}
Let $F=(\cont{S}{\mathbf{P}},Q)$ be a container in $Ind+1$ parameters with $S:\agdablue{Type}, \mathbf{P}\colon Ind\to S\to\agdablue{Type}$, $Q\colon S\to \agdablue{Type}$. For any fixed $\mathbf{X}\colon Ind\to \agdablue{Type}$, the type $\contfuncX{\agdablue{W}\;S\;Q}{\agdablue{Pos WAlg}}{\mathbf{X}}$ is the carrier of the initial algebra of $\llbracket F\rrbracket (\mathbf{X}, -) \colon\agdablue{Type}\to\agdablue{Type}$, i.e.\
\[\contfuncX{\agdablue{W}\;S\;Q}{\agdablue{Pos WAlg}}{\mathbf{X}} \cong \mu Y. \llbracket F\rrbracket (\mathbf{X}, Y).\]
\end{theorem}

\begin{proof}[Proof of \cref{thm-least-fp}] 
We write \agdablue{W} for $\agdablue{W}\;S\;Q$ and \agdablue{Posμ} for \agdablue{Pos WAlg}. We construct an $\llbracket F\rrbracket(\mathbf{X}, -)$-algebra with carrier $\contfuncX{\agdablue{W}}{\agdablue{Posμ}}{\mathbf{X}}$ by defining a morphism
\[\agdablue{into}\colon \llbracket F\rrbracket (\mathbf{X}, \contfuncX{\agdablue{W}}{\agdablue{Posμ}}{\mathbf{X}})\to \contfuncX{\agdablue{W}}{\agdablue{Posμ}}{\mathbf{X}}\]
by induction on \agdablue{Posμ} as follows.\footnote{We repackage the type of the input of \agdablue{into} via \cref{eqn:cont-func} and distributivity of functions over $\Sigma$. This also applies for the definition of \agdablue{out} in \cref{thm-greatest-fp}.}
\ExecuteMetaData[agda/latex/Code.tex]{morph}
Then $(\contfuncX{\agdablue{W}}{\agdablue{Posμ}}{\mathbf{X}}, \agdablue{into})$ is an $\llbracket F\rrbracket (\mathbf{X}, -)$-algebra. Now for any other algebra $(Y, \alpha)$, we need to define $\agdablue{$\overline\alpha$}\colon \contfuncX{\agdablue{W}}{\agdablue{Posμ}}{\mathbf{X}}\to Y$ uniquely such that the below diagram commutes.
\begin{equation}\label{comm-least}
\begin{tikzcd}
	\llbracket F\rrbracket(\mathbf{X}, \contfuncX{\agdablue{W}}{\agdablue{Posμ}}{\mathbf{X}})\arrow{r}{\agdablue{into}}\arrow[swap]{d}{\llbracket F\rrbracket(\mathbf{X}, \agdablue{$\overline\alpha$})} & \contfuncX{\agdablue{W}}{\agdablue{Posμ}}{\mathbf{X}}\arrow{d}{\agdablue{$\overline\alpha$}}\\
	\llbracket F\rrbracket (\mathbf{X}, Y)\arrow[swap]{r}{\alpha} &Y
\end{tikzcd}
\end{equation}

We define $\agdablue{$\overline\alpha$} \colon \underset{w : \agdablue{W}}{\sum} ((ind : Ind)\to \agdablue{Posμ}\; ind\; w \to \mathbf{X}\; ind)\to Y$ by induction on \agdablue{W}.\footnote{Technically, this definition raises a termination checking error, but this is easily fixed in the actual code by defining the uncurried version first then writing $\agdablue{$\overline\alpha$}$ in terms of it.}
\ExecuteMetaData[agda/latex/Code.tex]{ae2}
That \eqref{comm-least} commutes then follows definitionally.

The only thing left to show is that $\agdablue{$\overline\alpha$}$ is unique. We assume there is another arrow $\tilde{\alpha}\colon \contfuncX{\agdablue{W}}{\agdablue{Posμ}}{\mathbf{X}} \to Y$ making \eqref{comm-least} commute, i.e.\
\begin{equation}
\tilde{\alpha} \circ \agdablue{into} \;\agdablue{$\equiv$}\; \alpha \circ \llbracket F\rrbracket (\mathbf{X}, \tilde{\alpha}), \label{assumption1}
\end{equation}  
and prove that for $w : \agdablue{W}, k\colon \agdablue{Posμ}\; ind\; w \to \mathbf{X}\; ind$, we have $\tilde{\alpha}(w, k) \;\agdablue{$\equiv$}\; \agdablue{$\overline\alpha$} (w,k)$. By induction on \agdablue{W}, we have to show that for $s:S, f\colon Q\;s \to \agdablue{W}$, we have $\tilde{\alpha} (\agdagreen{sup-W}\;s\;f, k) \;\agdablue{$\equiv$}\; \agdablue{$\overline\alpha$} (\agdagreen{sup-W}\;s\;f, k)$. This follows easily from $\agdablue{$\overline\alpha$}$'s definition, assumption \eqref{assumption1}, and our inductive hypothesis.\qedhere

\end{proof}

Next, we show that $\contfuncX{\agdablue{M}\;S\;Q}{\agdablue{Pos MAlg}}{\mathbf{X}}$ is the terminal $\llbracket \cont{S}{\mathbf{P}}, Q \rrbracket (\mathbf{X}, -)$-coalgebra carrier. This proof is significantly more challenging than the previous one, both theoretically in that we use a modified version of the induction principle for \agdablue{Pos}, and also technically in that we have to go through several workarounds for \cubicalagda to accept our proof. It also requires us to use a considerable amount of path algebra to prove coherences that are not needed when assuming UIP, which appears to be implicitly assumed in the original proof. 

\begin{theorem}\label{thm-greatest-fp}
Let $F=(\cont{S}{\mathbf{P}},Q)$ be a container in $Ind+1$ parameters with $S:\agdablue{Type}, \mathbf{P}\colon Ind\to S\to\agdablue{Type}$, and $Q\colon S\to \agdablue{Type}$. For any fixed $\mathbf{X}\colon Ind\to \agdablue{Type}$, the type $\contfuncX{\agdablue{M}\;S\;Q}{\agdablue{Pos MAlg}}{\mathbf{X}}$ is the carrier of the terminal coalgebra of $\llbracket F\rrbracket (\mathbf{X}, -) \colon\agdablue{Type}\to\agdablue{Type}$, i.e.\
\[\contfuncX{\agdablue{M}\;S\;Q}{\agdablue{Pos MAlg}}{\mathbf{X}} \cong \nu Y. \llbracket F\rrbracket (\mathbf{X}, Y).\]
\end{theorem}
Before we prove~\cref{thm-greatest-fp}, we spell out what it is we need to show. We write $\agdablue{M}$ for $\agdablue{M}\;S\;Q$ and $\agdablue{Posν}$ for $\agdablue{Pos MAlg}$. First, we construct an $\llbracket F\rrbracket (\mathbf{X}, -)$-coalgebra with carrier $\contfuncX{\agdablue{M}}{\agdablue{Posν}}{\mathbf{X}}$ by defining
\[\agdablue{out}\colon \contfuncX{\agdablue{M}}{\agdablue{Posν}}{\mathbf{X}}\to \llbracket F\rrbracket (\mathbf{X}, \contfuncX{\agdablue{M}}{\agdablue{Posν}}{\mathbf{X}})\]
roughly as $\agdablue{into}^{-1}$, where \agdablue{into} is the function from \cref{thm-least-fp}. 
\ExecuteMetaData[agda/latex/Code.tex]{out}
So $(\contfuncX{\agdablue{M}}{\agdablue{Posν}}{\mathbf{X}}, \agdablue{out})$ is an $\llbracket F\rrbracket (\mathbf{X}, -)$-coalgebra. For any other coalgebra $(Y, \beta)$, we need to define $\agdablue{$\overline\beta$}\colon Y\to \contfuncX{\agdablue{M}}{\agdablue{Posν}}{\mathbf{X}}$ uniquely, such that the below diagram commutes.
\begin{equation}\label{comm-greatest}
\begin{tikzcd}
  Y\arrow[swap]{d}{\agdablue{$\overline\beta$}}\arrow{r}{\beta} &\llbracket F\rrbracket (\mathbf{X}, Y)\arrow{d}{\llbracket F\rrbracket(\mathbf{X}, \agdablue{$\overline\beta$})}\\
  \contfuncX{\agdablue{M}}{\agdablue{Posν}}{\mathbf{X}} \arrow[swap]{r}{\agdablue{out}}  &\llbracket F\rrbracket(\mathbf{X}, \contfuncX{\agdablue{M}}{\agdablue{Posν}}{\mathbf{X}})      
\end{tikzcd}
\end{equation}
To this end, from now on we fix $\beta\colon Y\to \llbracket F\rrbracket (\mathbf{X}, Y)$ with the following components.
\begin{align*}
\beta s&\colon Y\to S\\
\beta g&\colon (y : Y)\;(ind:Ind)\to \mathbf{P}\; ind\; (\beta s\; y) \to \mathbf{X}\; ind\\
\beta h&\colon (y : Y)\to Q\; (\beta s\; y) \to Y
\end{align*}
To prove~\cref{thm-greatest-fp} we (i) construct
$\agdablue{$\overline\beta$}\colon Y\to \underset{m : \agdablue{M}}{\sum}
((ind:Ind)\to \agdablue{Posν}\; ind\; m \to \mathbf{X}\; ind)$ such
that \eqref{comm-greatest} commutes and (ii) show that this
$\agdablue{$\overline\beta$}$ is unique. This will be the content
of \cref{lem-beta-exists,lem-beta-unique}.
\begin{lemma}\label{lem-beta-exists}
There is a function $\agdablue{$\overline\beta$}\colon Y\to \underset{m : \agdablue{M}}{\sum} ((ind : Ind)\to \agdablue{Posν}\; ind\; m \to \mathbf{X}\; ind)$ making \eqref{comm-greatest} commute.
\end{lemma}
\begin{proof}[Proof/construction]
We will define $\agdablue{$\overline\beta$}$ by
\ExecuteMetaData[agda/latex/Code.tex]{bb}
where $\agdablue{$\overline\beta_1$} \colon Y \to \agdablue{M}$  is defined by coinduction on \agdablue{M} and $\agdablue{$\overline\beta_2$} \colon (y:Y)\; (ind :Ind)\to \agdablue{Posν}\; ind\; (\overline\beta_1\; y) \to \mathbf{X}\; ind\;$ is defined by induction on \agdablue{Posν} as follows.\\
\begin{minipage}[t]{0.4\textwidth}
\minipageskip
\ExecuteMetaData[agda/latex/Code.tex]{bone}
\end{minipage}
\begin{minipage}[t]{0.50\textwidth}
\minipageskip
\ExecuteMetaData[agda/latex/Code.tex]{btwo}
\end{minipage}\\
This construction makes \eqref{comm-greatest} commute by definition.
\end{proof}

To show that $\agdablue{$\overline\beta$}$ is unique, we assume there is another arrow $\tilde{\beta}\colon Y \to  \contfuncX{\agdablue{M}}{\agdablue{Posν}}{\mathbf{X}}$ making the above diagram commute, i.e.\
\begin{equation}
\agdablue{out}\circ \tilde{\beta} \;\agdablue{$\equiv$}\; \llbracket F\rrbracket (\mathbf{X}, \tilde{\beta}) \circ \beta, \label{assumption2}
\end{equation}
then show that $\tilde{\beta} \;\agdablue{$\equiv$}\; \agdablue{$\overline\beta$}$. Naming $\tilde{\beta}$'s first and second projections $\tilde{\beta}_1$ and $\tilde{\beta}_2$, assumption \eqref{assumption2} can be split up into the paths shown below. We remark that all but the first one of these paths are dependent paths.
In what follows, we fix $y:Y$.
\[\comm{1}\;y \;\colon \agdapink{shape}\; (\tilde{\beta}_1\; y) \;\agdablue{$\equiv$}\; \beta s\; y\]
\vspace{-1.2cm} \ \\
\[\comm{2}\;y\; \colon \agdapink{pos}\; (\tilde{\beta}_1\; y)\; \agdablue{$\approxeq$}\; (\lambda\; q\to \tilde{\beta}_1\; (\beta h\; y\; q)) \]
\vspace{-.5cm} \ \\
\mycomment{dependent path over $(\lambda\; i \to Q\; (\comm{1}\;y\; i) \to \agdablue{M})$}
\vspace{-.3cm}
\[\comm{3}\;y\; \colon (\lambda\; ind\; p\to  \tilde{\beta}_2\; y\; ind\; (\agdagreen{here}\; p)) \; \agdablue{$\approxeq$}\; \beta g\; y \]
\vspace{-.5cm} \ \\
\mycomment{dependent path over $(\lambda\; i \to (ind : Ind) \to P\; ind\; (\comm{1}\; y\; i) \to \;X \;ind)$}
\vspace{-.3cm}
\[\comm{4} \;y\; \colon  (\lambda\; ind\; q\; b\to \tilde{\beta}_2\; y\; ind\; (\agdagreen{below}\; q\; b)) \; \agdablue{$\approxeq$}\; (\lambda\; ind\; q\; b\to \tilde{\beta}_2\; (\beta h\; y\; q)\; ind\; b)\]
\vspace{-.5cm} \ \\
\mycomment{dependent path over $(\lambda\; i\to (ind : Ind) (q : Q\; (\comm{1}\; y\; i)) \to \agdablue{Posν}\;ind\; (\comm{2}\; y\; i\; q) \to X\; ind)$}
\vspace{.2cm}\\
These equations simply express the fact that for $\tilde{\beta}$ to make diagram~\eqref{comm-greatest} commute, $\tilde{\beta}_1$ and $\tilde{\beta}_2$ have to be defined in the same way component-wise as \agdablue{$\overline\beta_1$} and \agdablue{$\overline\beta_2$}, up to a path.

\begin{lemma}\label{lem-beta-unique}
The function $\agdablue{$\overline\beta$}\colon Y\to \underset{m
: \agdablue{M}}{\sum} ((ind : Ind)\to \agdablue{Posν}\; ind\; m \to \mathbf{X}\;
ind)$ from \cref{lem-beta-exists} is unique. In other words, under the assumption of the existence of $\comm{1}$--$\comm{4}$ above, we can construct the following paths.
\textnormal{\ExecuteMetaData[agda/latex/Code.tex]{fstEqGoal}}
\vspace{-1cm}
\textnormal{\ExecuteMetaData[agda/latex/Code.tex]{sndEqGoal}}
\vspace{-.6cm} \ \\
\mycomment{\textnormal{dependent path over $(\lambda i \to (ind : Ind) \to \agdablue{Posv} \;ind\; (\agdablue{fstEq}\;y\;i)\; X\;ind )$}}       
\end{lemma}
\begin{proof}[Proof of \cref{lem-beta-unique}, part 1: construction of \agdablue{fstEq}]
Recall the coinduction principle \agdablue{MCoind} from \cref{sec-background}. Using this, we can prove \agdablue{fstEq} in a rather straightforward manner. To apply it, we need to construct a binary relation \agdablue{R} on $\agdablue{M}$. We construct it as an inductive family that relates precisely those terms we need to prove equal, i.e.\ $\tilde{\beta}_1\; y$ and $\agdablue{$\overline\beta_1$}\; y$.
\ExecuteMetaData[agda/latex/Code.tex]{R}
We then prove that it is a bisimulation using copattern matching.
\ExecuteMetaData[agda/latex/Code.tex]{isBisimR}
\vspace{-.6cm} \ \\
\mycomment{
Here, the second goal is of type $\agdablue{R}\;(\agdapink{pos}\;(\tilde{\beta}_1\;y\;q_0))\;(\overline{\beta}_1\;(\beta h \; y\; q_1))$ while $\agdagreen{R-intro}\;(\beta h \; y\; q_1)$ is of type $\agdablue{R}\;(\tilde{\beta}_1\;(\beta h \;y\;q_1))\;(\overline{\beta}_1\;(\beta h \; y\; q_1))$. This mismatch is adjusted using $\comm{2}$. Explicitly, we transport over the path of types $(\lambda \,i \to \agdablue{R}\; (\comm{2}\;y\;(\sim i)\; (\textit{q-eq} \; (\sim i))))$.
}
 \ \\
This allows us to finish the construction of \agdablue{fstEq}.\\
\begin{minipage}{0.75\textwidth}
\ExecuteMetaData[agda/latex/Code.tex]{fstEqBad}
\end{minipage}
\begin{minipage}{0.25\textwidth}
\hfill\qedhere
\end{minipage}
\end{proof}

Before we continue with the construction of \agdablue{sndEq}, let us
briefly discuss some of the finer points concerning the construction
of \agdablue{fstEq}. Because we used \agdablue{MCoind}
and \agdablue{isBisimR} to construct \agdablue{fstEq}, its definition
is somewhat opaque. Fortunately, the construction is well-behaved
on \agdapink{shape} and thus the action of 
$\agdapink{shape}$ on $(\agdablue{fstEq}\;y)$ computes
definitionally to $\comm{1}\;y$. This means that the action of \agdapink{pos} on $(\agdablue{fstEq}\;y)$  can be
viewed as a dependent path $\agdapink{pos} \;(\tilde{\beta}_1 \;y) \;\agdablue{$\approxeq$}\;\agdablue{$\overline\beta_1$} \;\agdablue{$\circ$}\; ({\beta h}\; y)$ over the path of types ($\lambda\;i\;\to\; Q \;(\comm{1} \;y \;i) \;→ \;\agdablue{M}$).
There is another canonical element of this type obtained by simply composing \comm{2} with a corecursive call of \agdablue{fstEq} --- let us call it \agdablue{fstEqPos}. It is defined as the composition of paths
\begin{equation}\label{eq-fstEqPos}
\begin{tikzcd}[ampersand replacement=\&]
	\agdapink{pos} \;(\tilde{\beta}_1\;y) \& {}  \& {(\lambda \;q \to \tilde{\beta}_1\;(\beta h \;y\;q))} \& {} \& {(\lambda \;q \to \agdablue{$\overline{\beta}_1$}\;(\beta h \;y\;q))}
	\arrow[squiggly, "comm_2\;y", from=1-1, to=1-3]
	\arrow["{\agdablue{fstEq}} \;\circ\; (\beta h \;y)", from=1-3, to=1-5]
\end{tikzcd}
\end{equation}
where the squiggly arrow indicates that $(comm_2\;y)$ is a dependent path.
We can now ask whether $\agdapink{pos}$ computes to
$(\agdablue{fstEqPos}\;y)$ on $(\agdablue{fstEq}\;y)$ (which in
essence just says that \agdablue{fstEq} satisfies the obvious
coinductive computational rule). This would be entirely trivial if we
had assumed UIP but now becomes something we cannot take for
granted. Fortunately, it turns out we can still prove it.
\begin{lemma}\label{lem-technical}
For all $y:Y$, we have $\agdablue{fstEqPos}\;y \;\agdablue{$\equiv$} \;(\lambda \;i \to \agdapink{pos}\;(\agdablue{fstEq}\;y\;i))$.
\end{lemma}
\begin{proof}[Proof sketch of \cref{lem-technical}]
The lemma is proved by abstracting and applying function
extensionality and path induction on \comm{1}. In this special case,
i.e.\ when $\comm{1}\;y = \agdablue{refl}$, one can simplify the
instances of \agdablue{isBisimR} and \agdablue{MCoind} used in the
construction of \agdablue{fstEq}. We omit the details which are just
technical path algebraic manipulations and refer the reader to the
formalisation.
\end{proof}
One may reasonably ask why this is a lemma and not simply
part of the \emph{definition} of \agdablue{fstEq}.
We discuss this in \cref{subsec:discussion}.

Let us now move on to the construction of \agdablue{sndEq}.
We construct \agdablue{sndEq} using~\cref{lem-pos-ind} which requires the following fact.
\begin{lemma}\label{lem-retract}
$\agdablue{$\overline{\beta}_1$}$ is an \agdablue{MAlg}-retraction.
\end{lemma}
\begin{proof}
By unfolding definitions, we see that the statement boils down to constructing, for each $y:Y$, a function $\widehat{f}_y : Q\;(\beta s\;y) \to Y$ such that the following identity holds for each $q : Q\;(\beta s\;y)$.
\begin{equation}\label{eq-retract}
\agdablue{$\overline{\beta}_1\;y$}\;(\widehat{f}_y\;q) \;\agdablue{$\equiv$}_{\agdablue{M}}\; \agdablue{$\overline{\beta}_1\;y$} \;(\beta h\;y\;q)
\end{equation}
Defining $\widehat{f}_y = \beta h\;y$ makes \eqref{eq-retract} hold
definitionally.
\end{proof}

Finally, we are ready to construct \agdablue{sndEq} and thereby
finish the proof of~\cref{lem-beta-unique}.
\begin{proof}[Proof of \cref{lem-beta-unique}, part 2: construction of \agdablue{sndEq}]
For ease of notation, we define, for each $ind : Ind$ and $y:Y$, the function $\agdablue{tr}\;y : \agdablue{Posν}\;ind\;(\agdablue{$\overline{\beta}_1$}\;y) \to \agdablue{Posν}\;ind\;(\tilde{\beta}_1\;y)$ to be the function obtained by transporting via $(\agdablue{fstEq}\;y)^{\agdablue{$-1$}}$.\footnote{Formally, we define $\agdablue{tr}\;y = \agdablue{transport}\;(\lambda \;i\to \agdablue{Posν}\;ind\;((\agdablue{fstEq}\;y)^{\agdablue{$-1$}}\;i))$.} 
For the construction of \agdablue{sndEq}, we first note that, by function extensionality and the interchangeability of dependent paths and transports, constructing \agdablue{sndEq} is equivalent to showing that
\begin{equation}\label{eq-sndEqAlt}
\tilde{\beta}_2\;y\;ind\; (\agdablue{tr}\;y\;{t}) \;\agdablue{$\equiv$}\; \agdablue{$\overline{\beta}_2$} \;y\;ind\;t
\end{equation}
for $ind : Ind$ and $t
: \agdablue{Posν}\;ind\;(\agdablue{$\overline{\beta}_1$}\;y)$. In light of \cref{lem-retract}, we
may apply \cref{lem-pos-ind} in order to induct on $t$. When $t$ is of the form $\agdagreen{here}\;p$, there is not much to
show. Indeed, by translating this instance of \eqref{eq-sndEqAlt}
back into the language of dependent paths, we see that the data is given
precisely by \comm{3}.

When $t$ is of the form $\agdagreen{below}\;q\;p$, we may assume inductively that we have a path
\begin{equation}\label{eq-indhyp}
\tilde{\beta}_2\;(\beta h \;y\;q)\;ind\; (\agdablue{tr}\;(\beta h \;y\;q)\;{p}) \;\agdablue{$\equiv$}\; \agdablue{$\overline{\beta}_2$} \;(\beta h \;y\;q)\;ind\;p
\end{equation}
and the goal is to show that
\begin{equation}\label{eq-goal}
\tilde{\beta}_2\;y\;ind \;(\agdablue{tr}\;y\;{(\agdagreen{below}\;q\;p)}) \;\agdablue{$\equiv$}\; \agdablue{$\overline{\beta}_2$} \;y\;ind\;(\agdagreen{below}\;q\;p).
\end{equation}

The RHS of \eqref{eq-goal} is, by definition, equal to the RHS of
\eqref{eq-indhyp}. By commuting transports with \agdagreen{below} and
using \comm{4}, we can rewrite the LHS of \eqref{eq-indhyp} to a term
of the form $\tilde{\beta}_2\;y\;ind \;(\agdagreen{below}
\;(\agdablue{transport}\;\agdablue{\dots}\;q)\;(\agdablue{transport}\;\agdablue{\dots}\;p))$.
Commuting transports with \agdagreen{below} in the LHS of
\eqref{eq-goal}, we get a term of the same form, albeit with
transports over slightly different families. Thus, it remains to equate
these families. We spare the reader the technical details and
simply point out that this task, after some path algebra, boils down
to precisely~\cref{lem-technical}. This concludes the proof of \cref{lem-beta-unique} and thus also of \cref{thm-greatest-fp}.
\end{proof}

\begin{example}
	For a concrete example of \cref{thm-least-fp,thm-greatest-fp}, consider \agdablue{S}, \agdablue{P$_0$}, and \agdablue{P$_1$} as defined in \cref{example:list-param}. Then for a fixed $X : \agdablue{Type}$, $\llbracket\cont{\agdablue{S}}{\agdablue{P$_0$}, \agdablue{P$_1$}}\rrbracket (X, -) \colon\agdablue{Type}\to\agdablue{Type}$ has an initial algebra with carrier $\contfuncX{\agdablue{$\mathbb{N}$}}{\agdablue{$\Fin$}}{X}$, and it has a terminal coalgebra with carrier
	 $\contfuncX{\agdablue{$\mathbb{N}\infty$}}{\agdablue{$\mathsf{Cofin}$}}{X}$. \agdablue{$\mathbb{N}\infty$} is defined as in \cref{ex:conats} and \agdablue{$\mathsf{Cofin}$} is the inductive type family over \agdablue{$\mathbb{N}\infty$} of finite and (countably) infinite sets. We note that $\llbracket\cont{\agdablue{S}}{\agdablue{P$_0$}, \agdablue{P$_1$}}\rrbracket (X, -) \cong \agdablue{$\top$} \uplus (- \times X)$, the signature functor for \agdablue{$\mathsf{List}$}. $\cont{\agdablue{$\mathbb{N}$}}{\agdablue{$\Fin$}}$ is the container representation of lists while $\cont{\agdablue{$\mathbb{N}\infty$}}{\agdablue{$\mathsf{Cofin}$}}$ is the container representation  of colists (the type of finite and infinite lists).
\end{example}

\subsection{The absence of UIP and Agda's termination checker}\label{subsec:discussion}

One of the key contributions of this paper is the fact we were
able to formalise~\cref{lem-beta-unique} without using UIP.
Our main difficulty was proving the technical \cref{lem-technical} which
essentially says that $\agdablue{fstEq}$ is coinductively defined in
the obvious manner.
In theory, \cubicalagda should allow us to define \agdablue{fstEq} as
\ExecuteMetaData[agda/latex/Code.tex]{fstEqGood}
where we recall that \agdablue{fstEqPos} is defined by coinductively calling \agdablue{fstEq} as in \eqref{eq-fstEqPos}.
This construction would make \cref{lem-technical} hold
definitionally, without requiring any form of UIP. There are, however, two issues with it.
Firstly, \agda does not accept this definition and
raises a termination checking error. We believe this to be an issue with
\cubicalagda's current termination checker. Generally speaking, in order to
check whether a corecursive function terminates, \agda needs to ensure
its output can be produced in a finite amount of steps. We call such
functions \textit{productive}. In the cases when it is not obvious from the
structure of the code that it is productive, e.g.\ if we make a
corecursive call and do something else with it before returning,
rather than returning directly, \agda usually raises a termination
error. While this is justified in general, composing productive calls
using \cubicalagda's primitive path composition function, \agdablue{hcomp}, should be productive, but \agda still raises an
error. This was raised as a GitHub issue \cite{github_issue}.

If the first issue were to be
resolved, our proof of \cref{thm-greatest-fp} could be made
significantly shorter, as we would not need to use \agdablue{M}'s coinduction
principle \agdablue{MCoind} in the definition of \agdablue{fstEq}.
Nevertheless, there is another issue with such a construction of \agdablue{fstEq}, namely that it relies heavily on the intricacies of cubical type theory (specifically when we introduce a path variable $i$ and then copattern match). As a result, we could not expect to reproduce such a proof in other non-cubical type theories. Therefore, from a mathematical perspective, the issue we faced with the termination checker may have been a blessing in disguise.
Our original motivation behind formalising \cref{lem-beta-unique} was to support the claim by Abbott et al.~\cite{abbott_tcs_2005}\ that the theory of containers can be understood in \emph{type theory}.
Here, we aim to interpret \emph{type theory} as generally as
possible, rather than restricting ourselves to cubical type theory.
Since we had to define  \agdablue{fstEq} using \agdablue{MCoind} rather than the \cubicalagda-specific
construction above, our proof should hold in any type theory having function extensionality (e.g.\ setoid type theory \cite{setoid-tt-19}) and coinduction for \agdablue{M}-types.
The fact that the authors of \cite{abbott_tcs_2005} never mention any results akin to \cref{lem-technical} suggests that they worked under a tacit assumption of UIP, which is further evidence that our formalisation indeed is a generalisation of previous results.

%% file: sections/conclusions.tex
%Contribution
In this paper, we presented a formalisation of the following results on containers, doing so without making any h-set assumptions, thereby lifting the original results from the category of sets to the wild category of types.

\begin{itemize}
\item $\contfuncX{\agdablue{W}\; S\; Q}{\agdablue{Pos WAlg}}{\mathbf{X}} \text{ is (the carrier of) the initial } \llbracket \cont{S}{\mathbf{P}} , Q\rrbracket (\mathbf{X}, -)\text{-algebra, and}$

\item $\contfuncX{\agdablue{M}\; S\; Q}{\agdablue{Pos MAlg}}{\mathbf{X}} \text{ is (the carrier of) the terminal } \llbracket \cont{S}{\mathbf{P}} , Q\rrbracket (\mathbf{X}, -)\text{-coalgebra.}$
\end{itemize}

While the first proof was straightforward, the second proof needed more careful consideration. In particular, it employed a modified version of \agdablue{Pos}'s induction principle and required various workarounds for \cubicalagda to accept our proof. These workarounds, however, suggested that our construction holds not only in \textit{cubical} type theory, but in any type theory with function extensionality and \agdablue{W-} and \agdablue{M-} types.

%Related work
Similar results have been formalised in \systemname{Lean} by Avigad et al.\ \cite{avigad-2019}, who utilised a variation on bounded natural functors in their formalisation, although to our knowledge our formalisation is the first one not relying on UIP. 
Vezzosi \cite{vezzosi-note} wrote about the semantics of allowing path abstraction in copattern matching definitions in \cubicalagda.
Ahrens et al.\ \cite{nonwellfounded-trees-2015} formalised \agdablue{M}-types in \cubicalagda as limits of chains, so that their definition does not use the \agdaorange{coinductive} keyword. Since our goal was not only to establish the theoretical results in a more general setting, but also to make use of \cubicalagda's support for proving equalities on coinductives, we chose to use native coinductive types and defined  \agdablue{M} as shown in \cref{sec:w-m}. It is, however, possible that using their definition might have circumvented the issues we faced with the termination checker, due to avoiding the use of native coinductive types, although their approach relies on the univalence axiom, while ours does not.

%Future work
The formalisation of \cref{thm-least-fp,thm-greatest-fp} is part of ongoing work on giving a comprehensive rendition and formalisation of containers in type theory. In terms of the main result of \cite{abbott_tcs_2005} stating that `each strictly positive type in $n$ variables can be interpreted as an $n$-ary container', we have proved that `container functors, without any h-level restriction, are closed under $\mu$ and  $\nu$'. In our experience, compared to strictly positive types formed using $0, 1, +, \times,$ and $\to$, $\mu$ and $\nu$ are the most challenging cases. One aspect of the original result in \cite{abbott_tcs_2005} that we are still missing is that they show that `\textit{containers} are closed under $\mu$ and $\nu$', i.e.\ containers, and not their functor interpretations, model strictly positive types. In the h-set level setting of \cite{abbott_tcs_2005}, one can easily move closure properties from container functors to containers via the fully faithful functor $\danascott\colon \mathsf{Cont}\to[\mathsf{I\to Set}, \mathsf{Set}]$. In our more general setting of using \agdablue{Type} instead of $\mathsf{Set}$, the wild functor on wild categories $\danascott\colon \mathsf{Cont}\to [\mathsf{I}\to\agdablue{Type}, \agdablue{Type}]$ being in some sense fully faithful does not immediately follow. To show this, we might need to somehow `tame' the wild category $[\mathsf{I}\to\agdablue{Type}, \agdablue{Type}]$ by adding coherences, but stating \agdablue{Type} as a higher category in HoTT is still an open problem (see e.g.\ \cite{Buchholtz2019}).

Having formalised these results on fixed points of containers without making h-set assumptions suggests that equivalent results should hold for the more general groupoid (or symmetric) containers and categorified containers \cite{hakon-masters, altenkirch-hott-uf-abstract}.  These more general kinds of containers are of interest as they can describe a larger class of types, such as the type of multisets, which are not covered by h-set based container definitions. Their theory is however not fully worked out yet, so we leave this for future work.